\documentclass{amsart}
\usepackage{amssymb, latexsym,amsmath,amsfonts}
\usepackage{xypic}

\def\ord{{\rm ord}}

\def\be{\begin{enumerate}}

\def\ee{\end{enumerate}}

\newtheorem{thm}{Theorem}[section]

\newtheorem{pro}[thm]{Proposition}

\newtheorem{lemma}[thm]{Lemma}

\newtheorem{de}[thm]{Definition}

\numberwithin{equation}{section}

\title{States of finite effect algebras}

\begin{document}

\author[G. Bi\'nczak]{G. Bi\'nczak$^1$}

\address{$^1$ Faculty of Mathematics and Information Sciences\\
Warsaw University of Technology\\
00-662 Warsaw, Poland}

\author[J. Kaleta]{J. Kaleta$^2$}

\address{$^2$ Department of Applied Mathematics\\
 Warsaw University of Life Sciences\\02-787 Warsaw, Poland}

\author[A. Zembrzuski]{A.Zembrzuski$^3$}

	\address{$^3$  Department of Informatics, Faculty of Applied Informatics and Mathematics,
	Warsaw University of Life Sciences\\02-787 Warsaw, Poland}

\email{$^1$ binczak@mini.pw.edu.pl, $^2$joanna\_kaleta@sggw.edu.pl, $^3$ andrzej\_zembrzuski@sggw.edu.pl}



\begin{abstract}

In this paper we  characterize finite effect algebras which have a state. We construct two matrices $A$ and $B$ assigned to a finite effect algebra $E$ and show that if $E$ has a state then rank$A=$ rank$B$.
\end{abstract}

\maketitle

\section{\bf Introduction}
Effect algebras have been introduced by Foulis and Bennet in 1994 (see \cite{FB94}) for the study of foundations of quantum mechanics (see \cite{DP00}). Independtly, Chovanec and K\^opka introduced an essentially equivalent structure called $D$-{\em poset} 
(see \cite{KC94}). Another equivalent structure was introduced by Giuntini and Greuling in \cite{GG94}).

The most important example of an effect algebra is $(E(H),0,I,\oplus)$, where $H$ is a Hilbert space and $E(H)$ consists of all self-adjoint operators $A$ on $H$ such that $0\leq A\leq I$. For $A,B\in E(H)$, $A\oplus B$ is defined if and only if $A+B\leq I$ and then $A\oplus B=A+B$. Elements of $E(H)$ are called {\em effects} and they play an important role in the theory of quantum measurements (\cite{BLM91},\cite{BGL95}).

A quantum effect may be treated as two-valued (it means $0$ or $1$) quantum  measurement that may be unsharp (fuzzy).  If there exist some pairs of effects $a,b$ which posses an orthosum $a\oplus b$ then this orthosum correspond to a  parallel measurement of two effects. 

In this paper we  assign to each finite effect algebra a  matrix $A$ such that $E$ has a state if and only if  rank $A$=rank $(A|1)$ and matrix equation 
$AX=\left[\begin{array}{c}1\\\vdots\\1\end{array}\right]$ has a  solution in the form
$X=\left[\begin{array}{c}x_1\\\vdots\\x_n\end{array}\right]$ such that $x_i\geq0$ for $1\leq i\leq n$.

Let us start with the following definition of an effect algebra.

\begin{de}
In \cite{FB94} an {\em effect algebra} is defined to be an algebraic system $(E,0,1,\oplus)$  consisting of a set $E$, two special elements $0,1\in E$ called  the {\em zero} and the {\em unit}, and a partially defined binary operation $\oplus$ on $E$ that satisfies  the following conditions for all $p,q,r\in E$:
\be
\item{} [Commutative Law] If $p\oplus q$ is defined, then $q\oplus p$ is defined and $p\oplus q=q\oplus p$.
\item {}[Associative Law] If $q\oplus r$ is defined and $p\oplus(q\oplus r)$ is defined, then $p\oplus q$ is defined, $(p\oplus q)\oplus r$ is defined, and $p\oplus(q\oplus r)=(p\oplus q)\oplus r$.
\item {}[Orthosupplementation Law] For every $p\in E$ there exists a unique $q\in E$ such that $p\oplus q$ is defined and $p\oplus q=1$.
\item {}[Zero-unit Law] If $1\oplus p$ is defined, then $p=0$.
\ee
\end{de}

For simplicity , we often refer to $E$, rather than to $(E,0,1,\oplus)$, as being an effect algebra.
\begin{de}
For effect algebras $E_1,E_2$ a mapping $\phi\colon E_1\to E_2$ is said to be an {\em isomorphism} if  $\phi$ is a bijection, 
$a\perp b\iff \phi(a)\perp\phi(b)$, $\phi(1)=1$ and $\phi(a\oplus b)=\phi(a)\oplus \phi(b)$.
\end{de}
If $p,q\in E$, we say that $p$ and $q$ are orthogonal and write $p\perp q$ iff $p\oplus q$ is defined in $E$. If $p,q\in E$ and $p\oplus q=1$, we call $q$ the {\em orthosupplement} of $p$ and write $p'=q$.

It is shown in \cite{FB94} that the relation $\leq$ defined for $p,q\in E$ by $p\leq q$ iff $\exists r\in E$ with $p\oplus r=q$ is a partial order on $E$ and $0\leq p\leq 1$ holds for all $p\in E$. It is also shown that the mapping $p\mapsto p'$ is an order-reversing involution and that $q\perp p$ iff $q\leq p'$. Furtheremore, $E$ satisfies the following {\em cancellation law}: If $p\oplus q\leq r\oplus q$, then $p\leq r$.

For $n\in{\mathbb N}$ and $x\in E$ let $nx=x\oplus x\oplus\ldots\oplus x$ ($n$-times). We write $\ord(x)=n\in {\mathbb N}$ if $n$ is the greatest integer such that $nx$ exists in $E$.

An {\em atom} of an effect algebra $E$ is a minimal element of $E\setminus\{0\}$. An effect algebra $E$ is {\em atomic} if for every non-zero element $x\in E$ there exists an atom $a\in E$ such that $a\leq x$.

\begin{de}
A state on an effect algebra $E$ is a map $s\colon E\to[0,1]\subseteq{\mathbb R}$ such that $s(1)=1$ and if $a\oplus b$ is defined, then $s(a)+s(b)\leq 1$ and $s(a\oplus b)=s(a)+s(b)$.
\end{de}

\begin{pro} \cite[Proposition 3.1]{W14}\label{W3.1}
Let $E$ be an orthocomplete atomic effect algebra. Then for every $x\in R$, there is a set $\{a_i|i\in I\}$ of mutually different atoms in $E$ and a set $\{k_i| i\in I\}$ of positive integers such that $x=\bigoplus\{k_ia_i|i\in I\}$.
\end{pro}

Let $E$ be a finite effect algebra. Then $A$ is orthocomplete and atomic. If $E$ has atoms $a_1,\ldots a_m$ then by  \ref{W3.1}  for every $x\in E$ there exist non-negative integers $k_1,\ldots,k_m$ such that $x=\bigoplus\limits_{i=1}^mk_ia_i$. By $M_{n\times n}({\mathbb N})$ we denote the set of all $n\times m$ matrices which quantities are natural numbers (including $0$).

\begin{de}
Let $$A=\left[\begin{array}{ccc}
y_{11}&\ldots&y_{1m}\\
\vdots&&\vdots\\
y_{n1}&\ldots&y_{nm}\end{array}\right]\in M_{n\times m}({\mathbb N})$$ then denote $$\left[\begin{array}{cccc}
y_{11}&\ldots&y_{1m}&1\\
\vdots&&\vdots&\vdots\\
y_{n1}&\ldots&y_{nm}&1\end{array}\right]$$ by $(A|1)$.
\end{de}

\begin{de}\label{matr}
Let $E$ be a finite effect algebra with atoms $a_1,\ldots,a_m$. Then define $Seq(E)=\{(x_1,\ldots,x_m)\in{\mathbb N}^m\colon\bigoplus\limits_{t=1}^m x_ta_t=1\}$. Let $n=card(Seq(E))$ and
$$
\begin{array}{l}
M(E)=\left\{\left[\begin{array}{ccc}
y_{11}&\ldots&y_{1m}\\
\vdots&&\vdots\\
y_{n1}&\ldots&y_{nm}\end{array}\right]\right.\in M_{n\times m}({\mathbb N})\colon\mathop{\forall}\limits_{1\leq i<j\leq n}\quad\mathop{\exists}\limits_{1\leq t\leq m}y_{it}\not=y_{jt}, \\
\{(y_{i1},\ldots,y_{im})\in{\mathbb N}^m\colon 1\leq i\leq n\}=Seq(E)\Bigg\}.\end{array}
$$
\end{de}

It turns out that  rows in $A\in M(E)$ are all elements of $Seq(E)$. Moreover if $A,B\in M(E)$ then $B$ can be obtained from $A$ by permuting  its rows.

In order to prove the main Theorem, we need use the following two lemmas:

\begin{lemma}\label{k<y}
Let $E$ be a finite effect algebra with atoms  $a_1,\ldots,a_m$ and $$A=\left[\begin{array}{ccc}
y_{11}&\ldots&y_{1m}\\
\vdots&&\vdots\\
y_{n1}&\ldots&y_{nm}\end{array}\right]\in M(E).$$ If $x=\bigoplus\limits_{t=1}^mk_ta_t\in E$ then there exists $1\leq i\leq n$ such that $k_t\leq y_{it}$ for $1\leq t\leq m$.
\end{lemma}
\begin{proof}
Let $E$ be a finite effect algebra with atoms  $a_1,\ldots,a_m$ and $$A=\left[\begin{array}{ccc}
y_{11}&\ldots&y_{1m}\\
\vdots&&\vdots\\
y_{n1}&\ldots&y_{nm}\end{array}\right]\in M(E).$$
Let  $x=\bigoplus\limits_{t=1}^mk_ta_t\in E$ and $x'=\bigoplus\limits_{t=1}^ml_ta_t$ then $\bigoplus\limits_{t=1}^m(k_t+l_t)a_t=x\oplus x'=1$ so
$(k_1+l_1,\ldots,k_m+l_m)\in Seq(E)$ and there exists $1\leq i\leq k$ such that $(k_1+l_1,\ldots,k_m+l_m)=(y_{i1},\ldots,y_{im})$. Hence
$k_t\leq k_t+l_t=y_{it}$ for $1\leq t\leq m$.
\end{proof}

\begin{lemma}\label{kl}
Let $E$ be a finite effect algebra with atoms  $a_1,\ldots,a_m$ and $$A=\left[\begin{array}{ccc}
y_{11}&\ldots&y_{1m}\\
\vdots&&\vdots\\
y_{n1}&\ldots&y_{nm}\end{array}\right]\in M(E).$$ If $x=\bigoplus\limits_{t=1}^mk_ta_t=\bigoplus\limits_{t=1}^ml_ta_t\in E$ then there exists $1\leq p,q\leq k$ such that $k_t\leq y_{pt}$, $l_t\leq y_{qt}$ and $y_{pt}-k_t=y_{qt}-l_t$ for $1\leq t\leq m$.
\end{lemma}
\begin{proof}
Let $E$ be a finite effect algebra with atoms  $a_1,\ldots,a_m$ and $$A=\left[\begin{array}{ccc}
y_{11}&\ldots&y_{1m}\\
\vdots&&\vdots\\
y_{n1}&\ldots&y_{nm}\end{array}\right]\in M(E).$$

Let $x=\bigoplus\limits_{t=1}^mk_ta_t=\bigoplus\limits_{t=1}^ml_ta_t\in E$. Let $z=x'=\bigoplus\limits_{t=1}^mr_ta_t$ then  $\bigoplus\limits_{t=1}^m(k_t+r_t)a_t=x\oplus x'=1$
and  $\bigoplus\limits_{t=1}^m(l_t+r_t)a_t=x\oplus x'=1$ so $(k_1+r_1,\ldots,k_m+r_m)\in Seq(E)$ and $(l_1+r_1,\ldots,l_m+r_m)\in Seq(E)$ and there exist $1\leq p,q\leq n$ such that
$(k_1+r_1,\ldots,k_m+r_m)=(y_{p1},\ldots,y_{pm})$ and $(l_1+r_1,\ldots,l_m+r_m)=(y_{q1},\ldots,y_{qm})$. Hence
$k_t\leq k_t+r_t=y_{pt}$, $l_t\leq l_t+r_t=y_{qt}$ and $y_{pt}-k_t=r_t=y_{qt}-l_t$ for $1\leq t\leq m$.
\end{proof}

The following theorem (the main theorem) gives the necessary and sufficient conditions enabling algebra to have a state.

\begin{thm}\label{main}
Let $E$ be a finite effect algebra with atoms  $a_1,\ldots,a_m$. Let $A\in M(E)$ and $B=(A|1)$. Then $E$ has a state if and only if rank $A=$ rank $B$
 and there exist
$s_1,\ldots,s_m\in[0,+\infty)$ such that $$A\cdot\left[\begin{array}{c}s_1\\\vdots\\s_m\end{array}\right]=\left[\begin{array}{c}1\\\vdots\\1\end{array}\right].$$
\end{thm}

\begin{proof} Let $$A=\left[\begin{array}{ccc}
y_{11}&\ldots&y_{1m}\\
\vdots&&\vdots\\
y_{k1}&\ldots&y_{km}\end{array}\right]\in M(E).$$
\be
\item[$\Leftarrow$]
Suppose that there exist $s_1,\ldots,s_m\in[0,+\infty)$ such that $$A\cdot\left[\begin{array}{c}s_1\\\vdots\\s_m\end{array}\right]=\left[\begin{array}{c}1\\\vdots\\1\end{array}\right].$$

Then $\sum\limits_{t=1}^my_{it}s_t=1$ for $1\leq i\leq k$.

Let $h\colon E\to{\mathbb R}$ be a map such that $h(\bigoplus\limits_{t=1}^mx_ta_t)=\sum\limits_{t=1}^mx_ts_t$. By \ref{W3.1} $h$ is defined for all $x\in E$.

 First we show that $h$ is well-defined.

Let $x\in E$ and $x=\bigoplus\limits_{t=1}^mk_ta_t=\bigoplus\limits_{t=1}^ml_ta_t$ by \ref{kl}  there exist $1\leq p,q\leq k$ such that $k_t\leq y_{pt}$, $l_t\leq y_{qt}$ and $y_{pt}-k_t=y_{qt}-l_t$ for $1\leq t\leq m$. Hence $\sum\limits_{t=1}^mk_ts_t-\sum\limits_{t=1}^ml_ts_t=\sum\limits_{t=1}^m(k_t-l_t)s_t=\sum\limits_{t=1}^m(y_{pt}-y_{qt})s_t=\sum\limits_{t=1}^my_{pt}s_t-\sum\limits_{t=1}^my_{qt}s_t=1-1=0$. Thus $\sum\limits_{t=1}^mk_ts_t=\sum\limits_{t=1}^ml_ts_t$ so $h(\bigoplus\limits_{t=1}^mk_ta_t)=h(\bigoplus\limits_{t=1}^ml_ta_t)$.

Now we show that $h\colon E\to[0,1]$. If $x=\bigoplus\limits_{t=1}^mk_ta_t$ then  there exists $1\leq i\leq n$ such that $k_t\leq y_{it}$ for $1\leq t\leq m$ by \ref{k<y}. Hence
$$
0\leq\sum_{t=1}^mk_ts_t=h(\bigoplus_{t=1}^m k_ta_t)\leq\sum_{t=1}^my_{it}s_t=1,
$$
that is why $0\leq h(x)\leq 1$.

Since $1=\bigoplus\limits_{t=1}^my_{1t}a_t$ so $h(1)=\sum\limits_{t=1}^my_{1t}s_t=1$.

Let $x_1,x_2\in E$, $x_1=\bigoplus\limits_{t=1}^mk_ta_t$ and $x_2=\bigoplus\limits_{t=1}^ml_ta_t$. Suppose $x_1\oplus x_2$ is defined then 
$$x_1\oplus x_2=\bigoplus_{t=1}^mk_ta_t\oplus\bigoplus_{t=1}^ml_ta_t=\bigoplus_{t=1}^m(k_t+l_t)a_t$$
 so by \ref{k<y} there exists  $1\leq i\leq n$ such that $k_t+l_t\leq y_{it}$ for $1\leq t\leq m$. Thus $h(x_1)+h(x_2)=\sum\limits_{t=1}^mk_ts_t+\sum\limits_{t=1}^ml_ts_t=\sum\limits_{t=1}^m(k_t+l_t)s_t\leq\sum\limits_{t=1}^my_{it}s_t=1$.

Moreover $h(x_1\oplus x_2)=h(\bigoplus\limits_{t=1}^m(k_t+l_t)a_t)=\sum\limits_{t=1}^m(k_t+l_t)s_t=\sum\limits_{t=1}^mk_ts_t+\sum\limits_{t=1}^ml_ts_t=h(x_1)+h(x_2)$. Hence
$h$ is a state.
\item[$\Rightarrow$]

Suppose that $E$ has a state $h\colon E\to[0,1]$. Let $s_t=h(a_t)\geq0$ for $1\leq t\leq m$. We know that $\bigoplus\limits_{t=1}^my_{it}a_t=1$ for $1\leq i\leq n$ since $A\in M(E)$.
Hence $1=h(1)=h(\bigoplus\limits_{t=1}^my_{it}a_t)=\sum\limits_{t=1}^my_{it}h(a_t)=\sum\limits_{t=1}^my_{it}s_t$. Thus
$$A\cdot\left[\begin{array}{c}s_1\\\vdots\\s_m\end{array}\right]=\left[\begin{array}{c}1\\\vdots\\1\end{array}\right].$$
By Rouch\'e-Capelli Theorem  rank $A=$ rank $B$, where $B=(A|1)$
\ee
\end{proof}

In \cite{G71} Greechie gives an example of finite effect algebra that has no states. This effect algebra $E$ has 12 atoms $\{a_1,\ldots,a_{12}\}$ such that $a_1\oplus a_2\oplus a_3\oplus a_4=1$,  $a_5\oplus a_6\oplus a_7\oplus a_8=1$,  $a_9\oplus a_{10}\oplus a_{11}\oplus a_{12}=1$,  $a_1\oplus a_5\oplus a_9=1$, $a_2\oplus a_6\oplus a_{10}=1$, $a_3\oplus a_7\oplus a_11=1$, $a_4\oplus a_8\oplus a_{12}=1$. Let

$$
A=\left[\begin{array}{cccccccccccc}
		1&	1&	1&	1&	0&	0&	0&	0&	0&	0&	0&	0\\
		0&	0&	0&	0&	1&	1&	1&	1&	0&	0&	0&	0\\
		0&	0&	0&	0&	0&	0&	0&	0&	1&	1&	1&	1\\
		1&	0&	0&	0&	1&	0&	0&	0&	1&	0&	0&	0\\
		0&	1&	0&	0&	0&	1&	0&	0&	0&	1&	0&	0\\
		0&	0&	1&	0&	0&	0&	1&	0&	0&	0&	1&	0\\
		0&	0&	0&	1&	0&	0&	0&	1&	0&	0&	0&	1
\end{array}\right]
$$

\bigskip
then $A\in M(E)$ and the echelon form of $(A|1)$ is (using Maxima)

\bigskip
$$\left[\begin{array}{ccccccccccccc}
		1&	0&	0&	0&	1&	0&	0&	0&	1&	0&	0&	0&	1\\
		0&	1&	0&	0&	0&	1&	0&	0&	0&	1&	0&	0&	1\\
		0&	0&	1&	0&	0&	0&	1&	0&	0&	0&	1&	0&	1\\
		0&	0&	0&	1&	0&	0&	0&	1&	0&	0&	0&	1&	1\\
		0&	0&	0&	0&	1&	1&	1&	1&	0&	0&	0&	0&	1\\
		0&	0&	0&	0&	0&	0&	0&	0&	1&	1&	1&	1&	1\\
		0&	0&	0&	0&	0&	0&	0&	0&	0&	0&	0&	0&	1
\end{array}\right].
$$

So rank$A=6$ and rank$(A|1)=7$. Thus $E$ has no state by \ref{main}.

We found (using Matlab) smaller example of finite effect algebra $E_4$ that has no states. This effect algebra $E_4$ has 3 atoms $\{a_1,a_2,a_3\}$ such that $a_1\oplus a_2\oplus a_3=1$,
$4a_1=4a_2=4a_3=1$. Let 

$$B=\left[\begin{array}{ccc}1&1&1\\4&0&0\\0&4&0\\0&0&4
\end{array}\right]
$$
then $B\in M(E_4)$ and rank$B=3$ and rank$(B|1)=4$ thus $E_4$ has no state by \ref{main}.

 The effect algebra $E_4$ is represented by the following Hasse diagram:

$$
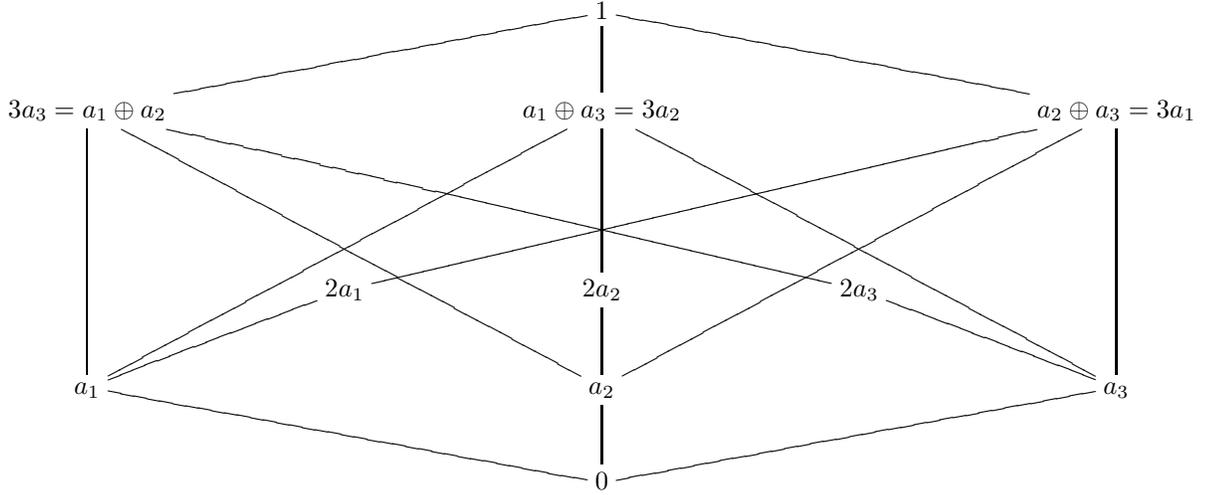
\xymatrix{&&&&1\ar@{-}[dllll]\ar@{-}[d]\ar@{-}[drrrr]&&&&\\
3a_3=a_1\oplus a_2\ar@{-}[ddd]\ar@{-}[dddrrrr]\ar@{-}[ddrrrrrr]&&&&a_1\oplus a_3=3a_2\ar@{-}[dddllll]\ar@{-}[dd]\ar@{-}[dddrrrr]&&&&a_2\oplus a_3=3a_1\ar@{-}[ddllllll]\ar@{-}[dddllll]\ar@{-}[ddd]\\
&&&&&&&&\\
&&2a_1\ar@{-}[dll]&&2a_2\ar@{-}[d]&&2a_3\ar@{-}[drr]&&\\
a_1\ar@{-}[drrrr]&&&&a_2\ar@{-}[d]&&&&a_3\ar@{-}[dllll]\\
&&&&0&&&&}
$$

It remains to check if in Theorem \ref{main} the condition that there exist
$s_1,\ldots,s_m\in[0,+\infty)$ such that $$A\cdot\left[\begin{array}{c}s_1\\\vdots\\s_m\end{array}\right]=\left[\begin{array}{c}1\\\vdots\\1\end{array}\right].$$ can be omitted.

Problem. Prove or disprove: if $E$ is a finite effect algebra with atoms  $a_1,\ldots,a_m$, $A\in M(E)$ and $B=(A|1)$ then $E$ has a state if and only if rank $A=$ rank $B$.

\end{document}